\documentclass[3p,twocolumn,UKenglish]{elsarticle}
\usepackage{microtype}
\usepackage{
  mathrsfs,
  amssymb,
  amsmath,
  amsthm,
  xspace,
  graphicx,
  enumerate,
  dsfont,
  hyperref,
  paralist
}

\theoremstyle{definition}
\newtheorem{theorem}{Theorem}
\newtheorem{lemma}{Lemma}
\newtheorem{corollary}{Corollary}
\newtheorem{proposition}{Proposition}

\newtheorem{definition}{Definition}
\newtheorem{claim}{Claim}

\newtheorem{remark}{Remark}

\newenvironment{itemize*}%
  {\vspace{-.5em}\begin{itemize}%
    \setlength{\itemsep}{0pt}%
    \setlength{\parskip}{0pt}}%
  {\end{itemize}\vspace{-.5em}}
\newenvironment{enumerate*}%
  {\begin{enumerate}%
    \setlength{\itemsep}{0pt}%
    \setlength{\parskip}{0pt}}%
  {\end{enumerate}}

\newcommand{\defstyle}[1]{\textbf{#1}}

\newcommand{\rowa}{r}
\newcommand{\cola}{c}
\newcommand{\rowb}{s}
\newcommand{\colb}{d}

\newcommand{\size}[1]{\# #1}
\newcommand{\len}[1]{|#1|}
\newcommand{\ie}{\textit{i.e.}}
\newcommand{\cf}{\textit{cf.}}
\newcommand{\eg}{\textit{e.g.}}

\newcommand{\set}[1]{\{#1\}}

\newcommand{\A}{\mathbb{A}}

\newcommand{\N}{\mathds{N}}
\newcommand{\Nz}{\mathds{N}_0}
\newcommand{\nexptime}{\textrm{\sc NExpTime}\xspace}
\newcommand{\np}{\textrm{\sc NP}\xspace}

\newcommand{\FOtwo}{\textup{FO}^{\textup{2}}}
\newcommand{\EMSOtwo}{\textup{EMSO}^2}
\newcommand{\suca}{{\rightarrow}}
\newcommand{\sucb}{{\downarrow}}

\title{Satisfiability for two-variable logic  \\with  two successor relations on finite linear orders\tnoteref{t1}} 
\tnotetext[t1]{Work supported by the FET-Open grant agreement FOX, number FP7-ICT-233599.}

\author{Diego Figueira} 
\ead{dfigueir@inf.ed.ac.uk}
\address{University of Edinburgh}

\begin{document}

\begin{abstract}
We study the finite satisfiability problem for first order logic with two variables and two binary relations, corresponding to the induced successor relations of two finite linear orders. We show that the problem is decidable in \nexptime.
 \end{abstract}

\maketitle

\section{Introduction}

First-order logic with two variables (henceforth denoted by $\FOtwo$) is of importance in computer science due to its decidable satisfiability problem (contrary to fragments of FO with 3 or more variables), and since it has connections with many formalisms, such as modal, temporal or description logics.
Many fragments of $\FOtwo$ have been studied because of this, especially in the presence of linear orders or equivalence relations. 
There are, still, a few relevant basic problems that remain open, and our work aims at expanding the classification of $\FOtwo$ in the presence of linear orders. In this setting, linear orders are related with temporal logics, but it is also applicable in other scenarios, like in databases or description logics.

\smallskip

We study the two variable fragment of first-order logic with two variables and two successor relations on two finite linear orders.  We show that the problem is decidable in \nexptime. This bound is optimal, since the problem is \nexptime-hard \cite{EtessamiVW02}. 
This logic has been previously claimed to be decidable in 2\nexptime in \cite{Manuel10}, but the proof was flawed.\footnote{In fact, according to its author, the proof of Lemma~4 in \cite{Manuel10} is wrong~\cite{ManuelPersonal}, and there does not appear to be an easy way of fixing it. This is a key lemma employed to obtain the decidability results contained in \cite{Manuel10}. (Of course, this does not affect the undecidability results contained in \cite{Manuel10}.) Here we adopt a different strategy to prove decidability.} 
As a corollary of the results from the report~\cite{manuelzeumeUNP}, this logic is shown to be decidable with a non-primitive-recursive algorithm.\footnote{By this we mean an algorithm whose time or space is not bounded by any primitive-recursive function.}
Our result also trivially extends to the satisfiability of existential monadic second order logic with two variables ($\EMSOtwo$) and two successor relations on finite linear orders.

This work focuses on the \emph{finite} satisfiability problem and hence all the results discussed next are relative to finite structures.
$\FOtwo$ is a well-known decidable fragment of first-order logic. Over arbitrary relations, it is known to be decidable~\cite{Mor75}, \nexptime-complete~\cite{GradelKV97}. $\FOtwo$ over words (\ie, with two relations: a successor relation over a finite linear order, and its transitive closure) is \nexptime-complete \cite{EtessamiVW02}. The satisfiability problem was shown to be undecidable:
 in the presence of two transitive relations (even without equality) the satisfiability problem is undecidable \cite{Kieronski05};
  in the presence of one transitive relation and one equivalence relation \cite{KieronskiT09};
 in the presence of three linear orders \cite{Kieronski11}; or
 in the presence of three equivalence relations \cite{KieronskiO05}.
However, if it has only two equivalence relations it is decidable \cite{KieronskiT09}. Over words with one equivalence relation it is decidable \cite{BDMSS10:tocl}. If it only has a transitive closure over a finite linear order and an equivalence relation, then it is \nexptime-complete \cite{BDMSS10:tocl}. If it only has a successor relation over a finite linear order and an equivalence relation, it is in 2\nexptime \cite{BDMSS10:tocl}. On trees, with only successor relations (\ie, the \emph{child} and  \emph{next sibling} relations) and an equivalence relation it is decidable in 3\nexptime \cite{BMSS09:xml:jacm}.

There have also been works in the presence of a linear order and a linear preorder 
\cite{SZlmcs12,manuelzeumeUNP}. In the presence of two finite linear orders, if there is a successor and its transitive closure over one linear order, and a successor over another linear order, it is decidable, and as hard as reachability of VAS according to the report \cite{manuelzeumeUNP}. If there are only two successors, it is known to be \nexptime-hard. Indeed, it is already \nexptime-hard even when no binary relations are present \cite{EtessamiVW02}. Here, we show that it is indeed \nexptime-complete. In fact, it sits in the same complexity class as $\FOtwo$ with just one successor relation on a linear order.

\section{Preliminaries}

Let $\N = \set{1, 2, \dotsc}$, $\Nz = \N \cup \set 0$, and for every $n,m \in \N$,  $[n,m] = \set{i\in \N \mid n \leq i \leq m}$, $[n] = [1,n]$.  Given a function $f : A \to B$ and a set $A' \subseteq A$, by $f|_{A'}$ we denote $f$ restricted to the elements of $A'$, and by $f[a \mapsto b] : A \cup \set a \to B \cup \set b$ we denote the function where $f[a \mapsto b] (a) = b$ and $f[ a \mapsto b] (a') = f(a')$ for all $a' \in A \setminus \set a$. We write $[a \mapsto b]$ to denote $\iota\,[a \mapsto b]$, where $\iota$ is the identity function, and $[a \mapsto b, a' \mapsto b']$ to denote $[a \mapsto b][a' \mapsto b']$.
  For any number $n$, we denote by $|n|$ its absolute value. We write $\size{S}$ to denote the number of elements of a set $S$. Given a string $w \in A^*$, we write $\len{w}$ to denote the length of $w$, 
 $w[i,j]$ to denote the subword of $w$ restricted to positions $[i,j]$, and $w[i]$  to denote $w[i,i]$,
for any $1 \leq i \leq j \leq \len{w}$.

\subsection{Permutations}

First-order structures with $n$ elements and two linear orders can be naturally represented as a permutation on $[n]$ with $n$ valuations.\footnote{We do not claim that we are the first to use this encoding, it may have been used before.} This representation will prove useful in the proofs that follow. 
A permutation of $[n]$ is represented as a set $\pi \subseteq [n] \times [n]$ so that for every $i \in [n]$, $\pi$ has exactly one pair with $i$ in the first component, and exactly one pair with $i$ in the second component. We will normally use the symbols $(\rowa,\cola), (\rowb,\colb)$ to denote elements of a permutation.

\begin{figure*}
  \centering
  \includegraphics[width=\textwidth]{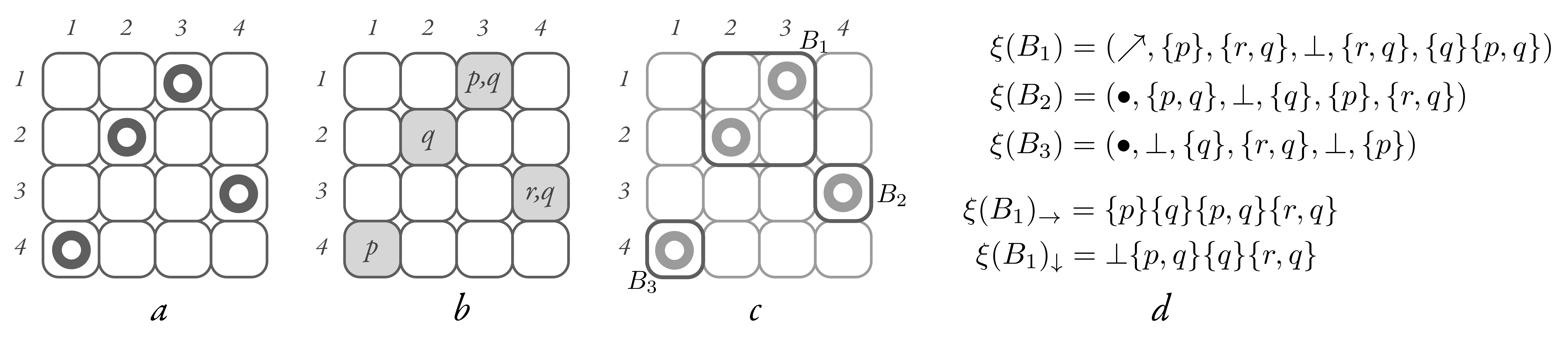}
  \caption{Representation of:
 ($a$) a $4$-permutation $\pi = \set{(4,1), (2,2), (1,3), (3,4)}$; ($b$) a valued permutation $(\pi,\sigma)$ where, \eg,  $\sigma(1,3) = \set{ p, q }$; ($c$) the maximal blocks of $\pi$; $(d)$ the fingerprints of $(\pi,\sigma)$.}
  \label{fig:valperm}
\end{figure*}

\newcommand{\ntype}[3]{[#2,#3]_{#1}}
Formally, given $n \in \N$, we say that $\pi \subseteq [n]\times[n]$ is an \defstyle{\textbf{\textit{n}}-permutation} if for every $k \in [n]$ we have $\size{\set{\cola \mid (k,\cola) \in \pi}} = \size{\set{\rowa \mid (\rowa,k) \in \pi}} = 1$. We say that $\pi$ is a permutation if it is an $n$-permutation for some $n$. Given a permutation $\pi$ and $(\rowa,\cola), (\rowa',\cola') \in \pi$ we say that the \defstyle{neighborhood type} of $(\rowa,\cola), (\rowa',\cola')$ in $\pi$ is an element 
$t \in \set{\bullet,\nearrow,\uparrow,\nwarrow,\leftarrow,\swarrow,\downarrow,\searrow,\rightarrow,\infty}$ so that:
$t \in \set{\searrow,\downarrow,\swarrow}$ if{f}
$\rowa' - \rowa = 1$,
$t \in \set{\nearrow,\uparrow,\nwarrow}$ if{f}
$\rowa' - \rowa = -1$,  $t \in \set{\nearrow,\rightarrow,\searrow}$ if{f}
$\cola' - \cola = 1$, $t \in \set{\nwarrow,\leftarrow,\swarrow}$ if{f}
$\rowa' - \rowa = -1$, and $t = \bullet$ if{f} $(\rowa,\cola) = (\rowa',\cola')$. We denote it with $\ntype{\pi}{(\rowa,\cola)}{(\rowa',\cola')}$. Figure~\ref{fig:valperm}-$a$ contains a graphical representation of a $4$-permutation, where $\ntype{\pi}{(2,2)}{(3,4)}=\downarrow$, $\ntype{\pi}{(2,2)}{(1,3)}=\nearrow$, $\ntype{\pi}{(1,3)}{(4,1)}=\infty$.

\newcommand{\V}{\mathbb{V}}
\newcommand{\nth}{\raisebox{1.3px}{\hspace{-0.4px}\text{\underline{\hspace{1.1ex}}\hspace{.15ex}}}}
Let us fix $\V$ to be an enumerable set of propositional letters. A \defstyle{valued permutation} is a pair $(\pi,\sigma)$ consisting of a permutation $\pi$ and a function $\sigma : \pi \to 2^\V$ that assigns a set of propositional letters to each  element of $\pi$. We say that $\sigma(r,c) \subseteq \V$ is the \defstyle{valuation} of $(r,c)$ in $(\pi,\sigma)$. Since for every $\rowa$ [resp.\ for every $\cola$] there is only one $\cola$ [resp.\ only one $\rowa$] such that $(\rowa,\cola) \in \pi$ we use the notation $\sigma(\rowa,\nth)$ [resp.\ $\sigma(\nth,\cola)$] to denote $\sigma(\rowa,\cola)$ for the only $\cola$ [resp.\ only $\rowa$] such that $(\rowa,\cola) \in \pi$. For example, the valued permutation of Figure~\ref{fig:valperm}-$b$ is such that $\sigma(3, \nth) = \sigma(3,4) = \set{r,q}$. A valued permutation can be seen as a finite first-order structure with two linear orders, where the permutation element $(\rowa,\cola)$ represents the $\rowa$-th element in the first linear order, and the $\cola$-th element in the second linear order, and its valuation is $\sigma(\rowa,\cola)$. Likewise, any finite first-order structure with two linear orders can be represented by a valued permutation.

For convenience in our proofs, we also represent an $n$-permutations with any set $\pi' \subseteq S \times T$ with $S, T \subseteq \N$, $\size S = \size T =n$ such that for every $(\cola,\rowa) \in S \times T$, $\size{\set{\cola' \mid (\rowa,\cola') \in \pi' }} = \size{\set{\rowa' \mid (\rowa',\cola) \in \pi' }} = 1 $. In this case we say that $\pi'$ is a \defstyle{permutation over $S \times T$}. Note that for every $n$-permutation $\pi'$ over $S \times T$ there is an $n$-permutation $\pi$ and a bijection $f : \pi' \to \pi$ that preserves the order of the elements: for all $f(\rowa,\cola) = (\rowa',\cola')$, $f(\rowb,\colb)=(\rowb',\colb')$, we have that $\rowa \leq \rowb$ if{f} $\rowa' \leq \rowb'$ and $\cola \leq \colb$ if{f} $\cola' \leq \colb'$. 
For any $(\rowa,\cola), (\rowa', \cola') \in \pi'$ we define $\ntype{\pi'}{(\rowa,\cola)}{(\rowa',\cola')}=\ntype{\pi}{f(\rowa,\cola)}{f(\rowa',\cola')}$.

Valued permutations represent, precisely, first-order finite structures with two linear orders. 
We define then the semantics of $\FOtwo(\suca, \sucb)$ over valued permutations.

\subsection{$\FOtwo$ with two linear orders}
We define $\FOtwo$ on finite structures with the induced successor relations of two finite linear orders, that we denote by $\FOtwo(\suca, \sucb)$. 
The atoms of $\FOtwo(\suca,\sucb)$ are: $p(a)$, $a \suca b$, and $a \sucb b$, for every $a, b \in \set{x,y}$ and every propositional letter $p \in \V$. If $\varphi, \psi$ are formulas of $\FOtwo$, so are $\exists a . \varphi$, $\forall a . \varphi$, $\lnot \varphi$, $\varphi \land \psi$, $\varphi \lor \psi$, where $a \in \set{x,y}$.
\newcommand{\vars}[1]{\+V_{#1}}
For any $\varphi \in \FOtwo(\suca, \sucb)$, let $\vars \varphi \subseteq \V$ be the set of all propositional variables occurring in $\varphi$.

\paragraph*{Semantics}
We define $\FOtwo(\suca,\sucb)$ on valued permutations. 
 The semantics are as expected, we give only some cases to fix notation. Here, $(\pi,\sigma)$ is a valued permutation and $\mu$ is a partial function $\mu : \set{x,y} \to \pi$.
{\small
   \begin{alignat*}{2}
     (\pi,\sigma) &\models_{\mu} p(x) && ~~\text{ if } p \in \sigma(\mu(x)) \\
     (\pi,\sigma) &\models_{\mu} \exists x . \varphi && ~~\text{ if for some $(\rowa,\cola) \in \pi$ we have } \\
&&&~~~(\pi,\sigma) \models_{\mu[x \mapsto (\rowa,\cola)]} \varphi\\
     (\pi,\sigma) &\models_{\mu} (x \suca y) && ~~\text{ if
       for some $(\rowa,\cola), (\rowa',\cola+1) \in \pi$ we} \\
&&&~~~\text{have $\mu(x) = (\rowa,\cola)$, $\mu(y) = (\rowa',\cola+1)$} \\
     (\pi,\sigma) &\models_{\mu} (x \sucb y) && ~~\text{ if for some
       $(\rowa,\cola), (\rowa+1,\cola') \in \pi$ we}\\
&&&~~~\text{have $\mu(x) =
       (\rowa,\cola)$, $\mu(y) = (\rowa+1,\cola')$}
   \end{alignat*}
}For any closed formula $\varphi$, we define $(\pi,\sigma) \models \varphi$ if $(\pi,\sigma) \models_{v_\emptyset} \varphi$, where $v_\emptyset(x) = v_\emptyset(y) = \bot$. In this case we say that $(\pi,\sigma)$ \emph{satisfies} $\varphi$. For example, the valued permutation of Figure~\ref{fig:valperm}-$b$ satisfies the formula $\forall x \forall y . \lnot ({x \suca y} \,\land\, {y \sucb x}\, \land\,  {p(x)})$. The \defstyle{satisfiability problem} for $\FOtwo(\suca,\sucb)$ is then, given a closed formula $\varphi \in \FOtwo(\suca,\sucb)$, whether $(\pi,\sigma) \models \varphi$ for some $(\pi,\sigma)$.

\paragraph*{Scott normal form}
Any formula $\varphi \in \FOtwo(\suca,\sucb)$ can be converted into a sa\-tis\-fia\-bility equivalent formula in Scott normal form, which is of the form
\begin{align*}
  \forall x \forall y ~ \chi ~~\land~~ \bigwedge_i \forall x \exists y ~
  \psi_i,
\end{align*}
where $\chi$ and all the $\psi_i$'s are quantifier-free formulas of $\FOtwo(\suca,\sucb)$. The resulting formula is linear in terms of the size of the original formula. Further, this reduction is polynomial-time (see, \eg,~\cite{GradelO99}). Henceforward we assume that all the formulas we work with are in Scott normal form, unless otherwise stated.

\section{Results}

\begin{theorem}\label{th:FOtwo(suca,sucb):nexptime}
  The satisfiability problem for $\FOtwo(\suca, \sucb)$   is \nexptime-complete.
\end{theorem}
As an immediate corollary we have that the same bound holds for $\EMSOtwo(\suca, \sucb)$, where $\EMSOtwo(\suca, \sucb)$ stands for formulas of $\FOtwo(\suca, \sucb)$ prefixed by existential quantification over sets of permutation elements.
\begin{corollary}
  The satisfiability problem for   $\EMSOtwo(\suca, \sucb)$   is \nexptime-complete.
\end{corollary}

\paragraph*{Proof sketch}

First, in Section~\ref{sec:few.small.blocks} we show a property of the blocks of a valued permutation. A block of a permutation can be seen as a set of positions $\set{(\rowa, \cola), (\rowa+1, \cola+1), \dotsc, (\rowa+k, \cola+k)}$ (or $\set{(\rowa, \cola), (\rowa+1, \cola-1), \dotsc, (\rowa+k, \cola-k)}$) of the permutation.\footnote{A similar notion of \emph{block} is also used in \cite{manuelzeumeUNP}.} We prove that if a $\FOtwo(\suca,\sucb)$ formula $\varphi$ is satisfiable, then it is satisfiable in a valued model where every block is of size bounded exponentially in the size of $\varphi$. Moreover, the number of different types of blocks that can appear in the valued permutation is also bounded exponentially in the size of $\varphi$.

Second, in Section~\ref{sec:permutatins.constraints} we show a combinatorial proposition. This involves what we call $n$-permutation \emph{constraints}, which are sets of positions of $[n] \times [n]$ where a permutation satisfying this constraint is not allowed to have an element. We give a sufficient condition on how large $n$ must be to ensure that there exists a permutation satisfying any constraint with a certain property---namely that it has at most 4 elements in any row or column.

Finally, in Section~\ref{sec:restricted.labeled.perm.pb} we introduce a problem called the Restricted Labeled Permutation problem (RLP), which we show to be decidable in $\np$ using the result of Section~\ref{sec:permutatins.constraints}. We then show that satisfiability for $\FOtwo(\suca, \sucb)$ can be reduced in $\nexptime$ to the RLP problem, by using the results on the size of the blocks of Section~\ref{sec:few.small.blocks}. Thus, decidability of the satisfiability problem for $\FOtwo(\suca,\sucb)$ follows, with a tight upper bound of $\nexptime$.

\section{Few and small blocks properties}
\label{sec:few.small.blocks}
\begin{definition}[Block] Given an $n$-permutation $\pi$ we say that $B\subseteq [n] \times [n]$ is a  \defstyle{block} of $\pi$ if $B = [i,i+k] \times [j,j+k]$ for some $k \in \Nz$ so that $i,j, i+k, j+k \in [n]$, and either
\begin{itemize*}
  \item $\size B =1$, and in this case we say that $B$ has type `$\bullet$', or, otherwise,
  \item for every $(\rowa,\cola), (\rowa+1,\cola') \in B \cap \pi$, we have $\cola' = \cola+1$, and in this case we say that $B$ has type `$\searrow$', or
  \item for every $(\rowa,\cola), (\rowa+1,\cola') \in B \cap \pi$, we have $\cola' = \cola-1$, and in this case we say that $B$ has type `$\nearrow$'.
\end{itemize*}
We say that $k$ is the \defstyle{size} of the block $B$. A block $B$ is \defstyle{maximal} if there is no  block $B'$ of $\pi$ with $B \subsetneq B'$. Figure~\ref{fig:valperm}-$c$ shows the three maximal blocks of a permutation, one with type $\nearrow$ and two with type $\bullet$.
\end{definition}

\newcommand{\bb}{\mathbf{b}}

\begin{proposition} \label{prop:small-blocks}
Any minimal valued permutation satisfying $\varphi \in \FOtwo(\suca,\sucb)$ is such that  every block is of size at most exponential in $\varphi$.
\end{proposition}
In fact, note that $\FOtwo(\suca, \sucb)$ on blocks is basically like $\FOtwo(\suca)$ (first order logic with a successor relation on a linear order), where we have the exponential length model property \cite{EtessamiVW02}. However, note that a block is within a context of other blocks, and special care must be taken in order to preserve all the elements that may be needed outside the block.

\newcommand{\type}{\textrm{type}}
\begin{definition}[Fingerprint]
  Given a maximal block $B = [k,k+n] \times [l,l+n]$ of a valued permutation $(\pi,\sigma)$, we define
$\xi(B)=(t, b^\suca_{-1},b^\suca_{+1},b^\sucb_{-1},b^\sucb_{+1}, a_0\dotsb a_{n})$, where
\begin{itemize*}
\item $t \in \set{\nearrow,\searrow,\bullet}$ is the type of $B$,
\item $b^\suca_{+1} = \sigma(k+n+1,\nth)$ if it exists, or $b^\suca_{+1} = \bot$ otherwise; $b^\sucb_{+1} = \sigma(\nth,l+n+1)$ if it exists, or $b^\sucb_{+1} = \bot$ otherwise,
\item $b^\suca_{-1} = \sigma(k-1,\nth)$ if it exists, or $b^\suca_{-1} = \bot$ otherwise; $b^\sucb_{-1} = \sigma(\nth,l-1)$ if it exists, or $b^\sucb_{-1} = \bot$ otherwise,
\item $a_i = \sigma(k+i,\nth)$ for all $0 \leq i \leq n$.
\end{itemize*}
$\xi(B)$ is the \defstyle{fingerprint} of $B$, and $t$ is the \defstyle{type} of $\xi(B)$ (notation: $\type(\tau)=t$, where $\xi(B)=\tau$). For $\xi(B) = \tau$, we also define $\tau_\suca = b^\suca_{-1} a_0 \dotsb a_n b^\suca_{+1}$; and $\tau_\sucb = b^\suca_{-1} a_0 \dotsb a_n b^\suca_{+1}$ if $t = \searrow$, or $\tau_\sucb = b^\suca_{-1} a_n \dotsb a_0 b^\suca_{+1}$ otherwise. The set of fingerprints of a valued permutation is the set of the fingerprints of all its maximal blocks. Figure \ref{fig:valperm}-$d$ contains an example of the fingerprints of a valued permutation.
\end{definition}

\begin{proposition}\label{prop:few-fingerprints}
  If $\varphi$ is satisfiable, then it is satisfiable in a minimal model with at most an exponential number of  fingerprints.
\end{proposition}

By Propositions \ref{prop:small-blocks} and \ref{prop:few-fingerprints}, we can restrict our attention to permutations with labels, over the exponential alphabet of fingerprints of maximal blocks.
However, to do this we need restrict the possible permutations. For example, there cannot be two elements $(\rowa,\cola), (\rowa+1,\cola+1)$ in the permutation where both its labels contain fingeprints with type $\searrow$. Indeed, this would imply that the blocks to which these fingerprint correspond were not actually maximal. This suggests that  we need to deal with some sort of \emph{constraints} defining  valid permutations. 
This is the theme of the following section.

\section{Permutations under constraints}
\label{sec:permutatins.constraints}

We define constraints that restrict where permutations may or may not contain elements. A constraint specify some positions in which a permutation satisfying it is not allowed to have an element.
More precisely, a $(n,k)$ constraint contains not more than $k$ forbidden positions in an $n$-permutation.

\begin{definition}[$(n,k)$-constraint]  
Given $n,k \in \N$, and $S,T \subseteq \N$ with $\size{S} = \size{T} = n$, we say that 
$\zeta \subseteq S \times T$ is a $(n,k)$-constraint over $S \times T$ if for every $(\rowa,\cola) \in S \times T$ we have $\size{\set{\cola' \mid (\rowa,\cola') \in \zeta}} \leq k$ and $\size{\set{\rowa' \mid (\rowa',\cola) \in \zeta}} \leq k$. An $n$-permutation $\pi$ over $S \times T$ satisfies a $(n,k)$-constraint $\zeta$ over $S \times T$ if $\pi \cap \zeta = \emptyset$. If $S = T = [n]$ we say that $\zeta$ is just a $(n,k)$-constraint.
\end{definition}

\begin{remark}
As with the permutations, any $n$-permutation $\pi$ over $S$ satisfying a $(n,k)$-constraint $\zeta$ can be equivalently seen as a $n$-permutation $\pi'$ satisfying a $(n,k)$-constraint $\zeta'$ and vice-versa.
\end{remark}

\begin{proposition}
  For every $(n,k)$-constraint $\zeta$ with $n > 2k$, there is an $n$-permutation $\pi$ satisfying $\zeta$.
\end{proposition}
\begin{proof}
  This can be shown by a simple application of Hall's Marriage Theorem~\cite{Hall35} (see also \cite[p.36]{DIE05B}).
Remember that Hall's theorem---in its finite, graph theoretic formulation---states that for any bipartite graph $G = (V_1 \cup V_2, E)$ with bipartite sets $V_1$ and $V_2$ of equal size, $G$ has a perfect matching if, and only if, every subset $S \subseteq V_1$ verifies $|N_G(S)| \geq |S|$. In the formulation, $N_G(S)\subseteq V_2$ is the  neighbourhood of $S$ in $G$ (\ie, the set of vertices adjacent to some vertex of $S$).

Let $\zeta$ be a $(n,k)$ constraint where $n > 2k$. Consider a bipartite graph $G = (V_r \cup V_c, E)$, where $V_r = \set{r} \times [n]$, $V_c = \set{c} \times [n]$. Vertices from $V_r$ represent rows and vertices from $V_c$ represent a columns.
The set of edges $E \subseteq V_r \times V_c$ is defined as all pairs $((r,i),(c,j))$ so that $(i,j) \not\in \zeta$ (\ie, they represent  permutation positions that do not interefere any constraint). Hence, there is a perfect matching between $V_r$ and $V_c$ if, and only if, there is an $n$-permutation satisfying $\zeta$.
To prove that there is such a matching, by Halls' theorem it suffices to verify $|N_G(S)| \geq |S|$ for every $S \subseteq V_r$. We show this by case distinction. 
\vspace{-1em}
\begin{itemize*}
\item Suppose first $|S| \leq n-k$. Note that every vertex of $V_r$ has at least $n-k$ edges because there are only $k$ constraints in $\zeta$. Then, $N_G(S)$ has $n-k$ vertices because every single vertex in $S$ has already $n-k$ edges.
\item Suppose now $|S| > n-k$. Since $n > 2k$, we have $n-k > k$. Then, every vertex from $V_c$ already has one neighbor in $S$, as every vertex from $V_c$ has at least $n-k > k$ neighbors.
\end{itemize*}
Hence, there is a perfect matching, and thus there exists a permutation satisfying $\zeta$.
\end{proof}
We then have the following corollary.
\begin{corollary}\label{cor:constraint-(n,4)}
For every $(n,4)$-constraint $\zeta$ with $n\geq 9$, there is an $n$-permutation $\pi$ that satisfies $\zeta$.  
\end{corollary}

\section{Labeled permutations}
\label{sec:restricted.labeled.perm.pb}

In this section we prove the \nexptime upper bound of the satisfiability problem for $\FOtwo(\suca, \sucb)$, using the developments of the two previous sections.
The idea is to guess the (exponentially many) blocks (of exponential size) of a minimal model that satisfies $\varphi$, and use them as letters of our alphabet. Using this guessing, we reduce the satisfiability problem into a problem we introduce next, the \emph{Restricted Labeled Permutation problem} (RLP). 

\subsection{Restricted labeled permutation problem}
\begin{definition}
  A \defstyle{labeled permutation} over a (finite) alphabet $\A$ is a pair $(\pi, \lambda)$ where $\pi$ is a permutation and $\lambda : \pi \to \A$. Note that $(\pi,\lambda)$ is nothing else than a valued permutation where exactly one propositional variable holds at any position.
\end{definition}
\newcommand{\proj}{\textit{Proj}}
\begin{definition}[$(\pi, \lambda)_\suca$, $(\pi, \lambda)_\sucb$]
Given a labeled $n$-permutation $(\pi, \lambda)$ we define $(\pi, \lambda)_\suca$ and $(\pi, \lambda)_\sucb$ as follows:  $(\pi, \lambda)_\suca = \lambda(1,\cola_1) \dotsb \lambda(n,\cola_n)$ and $(\pi, \lambda)_\sucb = \lambda(\rowa_1,1) \dotsb \lambda(\rowa_n,n)$,
where $\set{(1,\cola_1), \dotsc, (n,\cola_n)}  = \set{(\rowa_1,1), \dotsc, (\rowa_n,n)} = \pi$.
\end{definition}

\begin{definition}
  A \defstyle{label restriction} over an alphabet $\A$ is a triple $(a,t,b)$, where $a,b \in \A$ and $t \in \set{\nearrow, \searrow}$. We say that a labeled permutation $(\pi, \lambda)$ \emph{satisfies} $(a,t,b)$ if for every $(\rowa,\cola), (\rowb,\colb) \in \pi$ such that $\lambda(\rowa,\cola) = a$, $\lambda(\rowb,\colb)=b$ we have $\ntype{\pi}{(\rowa,\cola)}{(\rowb,\colb)} \neq t$.
We say that $(\pi, \lambda)$ satisfies a \emph{set} of label restrictions if it satisfies all of its members.
\end{definition}

We define the main problem of this section.
Using the result of Section~\ref{sec:permutatins.constraints} on permutations under constraints, we show that this problem is in \np.

\begin{center}
\small
  \begin{tabular}{|rl|}
    \hline
\textsc{Problem:}& The Restricted Labeled Permutation \\
&problem (RLP)
\\
    \hline
    \textsc{Input:} & A finite alphabet $\A$, \\
&a set of label restrictions $R$, and\\
&two regular languages $\+L_1, \+L_2 \subseteq \A^*$ \\
&(given as NFA).
\\
    \textsc{Question:} & Is there a labeled permutation $(\pi, \lambda)$ \\
&satisfying $R$ 
such that \\
&$(\pi, \lambda)_\suca \in \+L_1$ and $(\pi, \lambda)_\sucb \in \+L_2$?\\
    \hline
  \end{tabular}
\end{center}

\newcommand{\lang}[1]{L(#1)}
\begin{proposition}\label{prop:rlpp:NP}
  The RLP problem is in \np.
\end{proposition}
\begin{proof}
\begin{figure*} \newcommand{\boxdepic}{\raisebox{-2.5pt}{\includegraphics[height=10.5px]{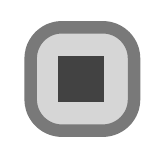}}}
\newcommand{\zonedepic}{\raisebox{-3.2pt}{\includegraphics[height=10.5px]{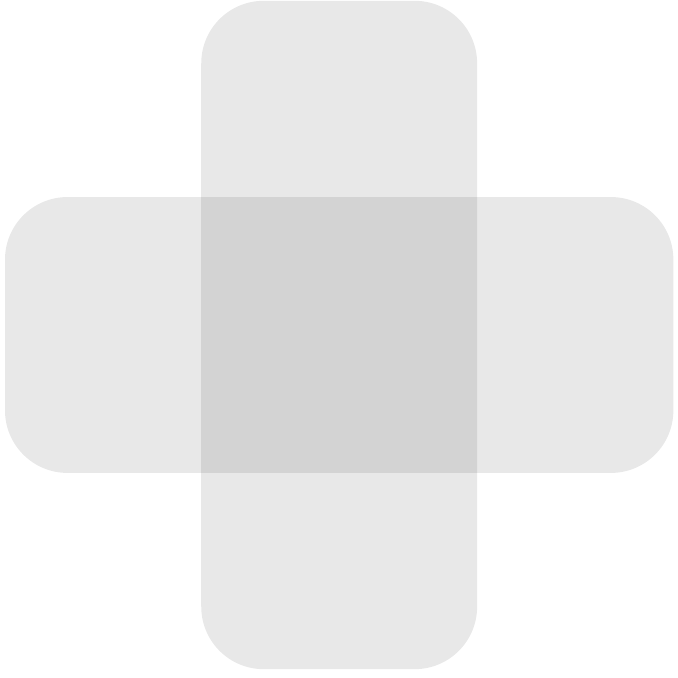}\hspace{0.8pt}}}
  \centering
  \includegraphics[width=\textwidth]{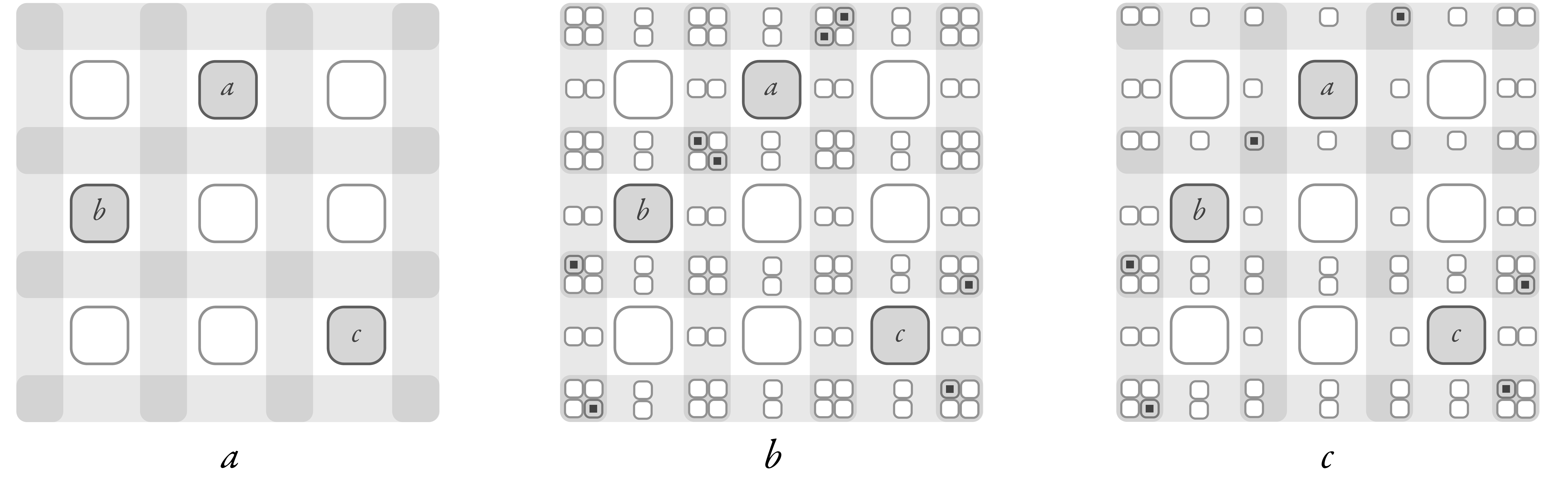}
  \caption{Example of the labeled permutation \ensuremath{(\pi',\lambda')} over the alphabet \ensuremath{\set{a,b,c, \Box}}, showing the zones (depicted as \protect\zonedepic{}) where the label \ensuremath{\Box} (depicted as \protect\boxdepic{}) can appear. Here, \ensuremath{g(a)=g(b)=g(c)=1}.}
  \label{fig:claim-polyboxes}
\end{figure*}
  Let $\+A_1, \+A_2$ be two NFA over the alphabet $\A$ corresponding to the regular languages $\+L_1, \+L_2$ respectively. Let $R$ be a set of restrictions.

The algorithm first guesses some properties of the labeled permutation $(\pi, \lambda)$ that  satisfies $R$ and is such that $(\pi, \lambda)_\suca \in \+L_1$ and $(\pi, \lambda)_\sucb \in \+L_2$. (We cannot simply guess $(\pi,\lambda)$ because it may be too big.) 
For each letter $a \in \A$ we guess if it appears exactly $k$ times in $(\pi,\lambda)$ for some $k \leq 17$, or if it appears more than $17$ times. Let $g : \A \to \set{0,1, \dotsc, 17, \infty}$ be this guessing; and let us define $\A_\leq = \set{a \in \A \mid g(a) \neq \infty}$ and $\A_> = \set{a \in \A \mid g(a) = \infty}$.

Let us call \emph{zone} to any set $S=[i,i+l] \times [j,j+l]$ for $i,j,l \in \N$.
Next, we guess a \emph{small} labeled permutation $(\pi',\lambda')$ over the alphabet $\A_\leq \cup \set{\Box}$. 
This labeled permutation is such that:
\begin{enumerate*}
  \item \label{it:atmost17} All the letters $a \in \A_\leq$ in $(\pi',\lambda')$ appear exactly $g(a)$ times.
  \item \label{it:notoomanybox}There is no zone $S$ so that
    \begin{itemize*}
    \item $\pi' \cap S$ has at least two elements, and
    \item $\lambda'(\rowa,\cola)=\Box$ for all $(\rowa,\cola) \in \pi' \cap S$.
    \end{itemize*}
  \item \label{it:pi'lam'-sats-R} $(\pi',\lambda')$ satisfies the restrictions $R$. 
\end{enumerate*}
\begin{claim}
All possible $(\pi',\lambda')$ satisfying the conditions above are labeled permutations of size polynomially bounded by $\size\A$.
\end{claim}
\begin{proof}
Note that, once we fix $g$, there are not more than
$
N=(~ 1 + \sum_{\substack{a \in \A\\g(a) \neq \infty}} g(a)~)^2
$ 
different zones containing only $\Box$ labels, that cover all positions where the label $\Box$ can occur in $(\pi',\lambda')$. These are the zones defined in between the elements of $\set{a \in \A \mid g(a) \neq \infty}$. Then, there cannot be more than $N$ elements with label $\Box$, since otherwise there would be at least one zone with more than one element, contradicting condition \eqref{it:notoomanybox}. Since $N \leq (1 + 17 \cdot \size\A)^2$, the claim follows.
For example, Figure~\ref{fig:claim-polyboxes}-$a$ depicts the $16=(1+ g(a)+g(b)+g(c))^2$ possible zones where the label $\Box$ can appear as the dark gray areas. In Figure~\ref{fig:claim-polyboxes}-$b$ we see that there is one zone (in fact, two) that contains more than one element $\Box$, and therefore condition \eqref{it:notoomanybox}  is falsified (for instance, when $S=[4,5] \times [4,5]$). Finally, Figure~\ref{fig:claim-polyboxes}-$c$ shows a labeled permutation $(\pi',\lambda')$ satisfying condition \eqref{it:notoomanybox}.
\end{proof}

\newcommand{\rex}{\textit{e}}
Let $\+A_{>}$ be an NFA over $\A$ that accepts all words $w \in \A^*$ such that  every $a \in \A_>$ appears more than 17 times in $w$.
Let $\rex_1$ [resp.\ $\rex_2$] be the regular expression resulting from replacing every appearance of $\Box$ in $(\pi',\lambda')_\suca$ [resp.\ in $(\pi',\lambda')_\sucb$] with the expression $(\A_>)^+$. 
Notice that, for every $i \in \set{1,2}$, any word $w$ of $\lang{\rex_i}\cap \lang{\+A_{>}}$ is such that  the number of appearances of $a \in \A_\leq$ in $w$ is exactly $g(a)$, and every other letter $a \in \A_>$ appears more than 17 times.
Let $\+A'_i$ denote the NFA corresponding to $\lang{\+A_i} \cap \lang{\rex_i}\cap \lang{\+A_{>}}$, for every $i \in \set{1,2}$. Observe that $\+A_{>}$, $\+A'_1$ and $\+A'_2$ can be built in polynomial time. Given a language $\+L \subseteq \A^*$, let $pk(\+L)$ denote the Parikh image of $\+L$.
We finally  check whether 
\[
pk(\lang{\+A'_1}) ~\cap~ pk(\lang{\+A'_2}) ~\neq~ \emptyset.
\]
This can be verified in \np  by computing the existential Presburger formulas for both automata in polynomial time \cite{VermaSS05}
 and checking for emptiness of its intersection in \np \cite{Papadimitriou:1981}.

 \begin{claim}
   $pk(\lang{\+A'_1}) ~\cap~ pk(\lang{\+A'_2}) ~\neq~ \emptyset$ if, and only if, there is a labeled permutation $(\pi, \lambda)$ that satisfies $R$, such that $(\pi, \lambda)_\suca \in \+L_1$ and $(\pi, \lambda)_\sucb \in \+L_2$.
 \end{claim}
The rest of the proof is dedicated to prove the statement above.

\smallskip

$[\Rightarrow]$
If $pk(\lang{\+A'_1}) ~\cap~ pk(\lang{\+A'_2}) \neq \emptyset$,  we show that there is a labeled permutation $(\pi, \lambda)$ that satisfies $R$ and such that $(\pi, \lambda)_\suca \in \+L_1$ and $(\pi, \lambda)_\sucb \in \+L_2$.

Let $w_1 \in \lang{\+A'_1}$, $w_2 \in \lang{\+A'_2}$ such that $pk(w_1) = pk(w_2)$, and let $m = |w_1| = |w_2|$. For every $a \in \A$, let $X_a \subseteq [m] \times [m]$ be defined as all $(\rowa,\cola) \in [m] \times [m]$ such that $w_1[\rowa]=a$ or $w_2[\cola] = a$ (\ie, $X_a$ is the set of possible permutation elements labeled with $a$). 
For any $A \subseteq \A$, let $X_{A} =  \bigcup_{a \in A} X_a$.

We define the labeled permutation $(\pi_{\A_\leq}, \lambda_{\A_\leq})$ as all the elements $(\rowa+\rowa', \cola+\cola')$  such that
\begin{itemize*}
\item $(\rowa,\cola) \in \pi'$, $\lambda'(\rowa,\cola) \neq \Box$, and
\item $\rowa'$ [resp.\ $\cola'$] is $k -  \ell$, where
  \begin{itemize*}
  \item $k$ is the number of occurrences of letters from $\A_>$ in $w_1$ [resp.\ in $w_2$] before the $\rowa$-th
    [resp.\ $\cola$-th] appearance of a letter from $\A_\leq$, and
  \item $\ell$ is the number of letters $\Box$ in $(\pi', \lambda')_\suca
    [1, \rowa]$ [resp.\ in $(\pi', \lambda')_\sucb [1, \cola]$].
  \end{itemize*}
Note that $k \geq \ell$.
\end{itemize*}
We define $\lambda_{\A_\leq}(\rowa+\rowa', \cola+\cola') = \lambda'(\rowa,\cola)$. We have that $(\pi_{\A_\leq},\lambda_{\A_\leq})$ is a labeled permutation over $X_{\A_\leq}$ satisfying $R$, as $(\pi',\lambda')$ satisfies $R$ by  \eqref{it:pi'lam'-sats-R}. In fact, it is equivalent to $(\pi',\lambda')$ when restricted to elements with labels in $\A_\leq$ (\cf~Figures~\ref{fig:construction-from-w1w2}-$a$, \ref{fig:construction-from-w1w2}-$b$).
\begin{figure*}
  \centering
  \includegraphics[width=\textwidth]{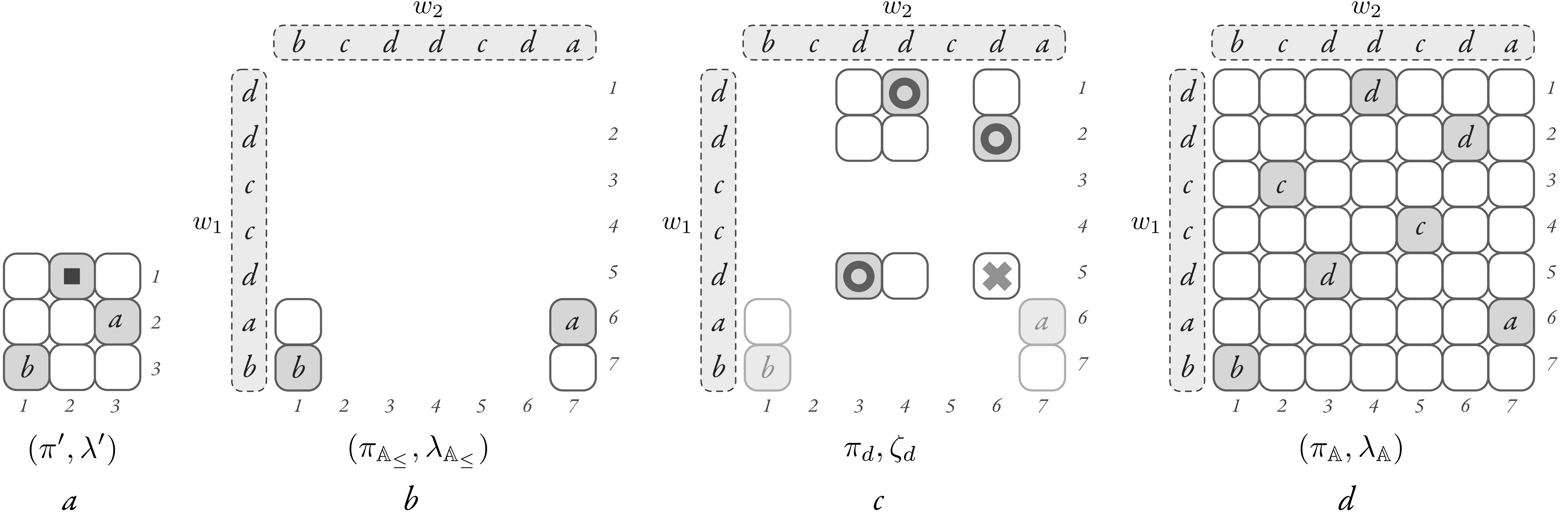}
  \caption{Example where $\A_\leq = \set{a,b}$, $\A_> = \set{c,d}$, $w_1 = ddccdab$, $w_2 = bcddcda$, $m=7$.
For illustration purposes, we let the threshold defining $\A_\leq$ and $\A_>$ to be $1$ instead of $17$.}
  \label{fig:construction-from-w1w2}
\end{figure*}
We build, for every $H \subseteq \A_>$, a labeled permutation $(\pi_{\A_\leq \cup H}, \lambda_{\A_\leq \cup H})$ over $X_{\A_\leq \cup H}$ satisfying $R$, so that $(\pi_{\A_\leq \cup H}, \lambda_{\A_\leq \cup H})_\suca$ [resp.\ $(\pi_{\A_\leq \cup H}, \lambda_{\A_\leq \cup H})_\sucb$] is $w_1$ [resp.\ $w_2$] projected onto $\A_\leq \cup H$. Note that when $H = \A_>$, $(\pi_{\A_\leq \cup H}, \lambda_{\A_\leq \cup H})$ is the labeled permutation over $X_\A = [m]\times [m]$ we are looking for. We build these inductively.

The base case is when $H=\emptyset$ and we then have $(\pi_{\A_\leq},\lambda_{\A_\leq})$, which clearly satisfies $R$. Suppose now we have constructed $(\pi_{\A_\leq \cup H}, \lambda_{\A_\leq \cup H})$ for some $H$, and let $a \in \A_>$ be such that $a \not\in H$. 
Let $\zeta_a$ be the set of all $(\rowa,\cola) \in X_a$ such that there is some $(\rowa',\cola') \in \pi_{\A_\leq \cup H}$ with $|\rowa-\rowa'|=|\cola-\cola'|=1$.
Using Corollary~\ref{cor:constraint-(n,4)} one can show that there is always a permutation over $X_a$ satisfying $\zeta_a$, so that it does not have any two elements one next to the other.  In the example of Figure~\ref{fig:construction-from-w1w2}, we see in item $c$ an illustration of a possible such permutation over $X_d$, and the constraints originating in this case from $(\pi_{\A_\leq}, \lambda_{\A_\leq})$.)

\begin{claim}
  There is a permutation $\pi_a$ over $X_a$ satisfying the constraints $\zeta_a$ such that there are no two $(\rowa,\cola), (\rowa',\cola') \in \pi_a$ with $|\rowa-\rowa'| = |\cola-\cola'|=1$.
\end{claim}
\begin{proof}
We first partition $X_a$ into two sets $X'_a$, $X''_a$ so that there are no two elements in $X'_a$ (resp.\ $X''_a$) with neighboring rows or columns.
Let $X'_a$ be the set of all $(\rowa,\cola) \in X_a$ such that $\size{\set{\rowa' \mid \rowa'\leq \rowa \land (\rowa',\cola) \in X_a}}$ and $\size{\set{\cola' \mid \cola'\leq \cola \land (\rowa,\cola') \in X_a}}$ are \emph{odd}; and let $X''_a$ be the set of all $(\rowa,\cola) \in X_a$ such that $\size{\set{\rowa' \mid \rowa'\leq \rowa \land (\rowa',\cola) \in X_a}}$ and $\size{\set{\cola' \mid \cola'\leq \cola \land (\rowa,\cola') \in X_a}}$ are \emph{even}. Since $\size{X_a} \geq 18^2$, we have that $\size{X'_a} \geq 9^2$ and $\size{X''_a} \geq 9^2$. Note that $\zeta_a$ does not have more than $4$ restrictions on each row and on each column. Then, $\zeta_a \cap X'_a$ is an $(\ell,4)$-constraint with $\ell \geq 9$. Hence, by Corollary~\ref{cor:constraint-(n,4)}, there exists a permutation $\pi'_a$ over $X'_a$ that satisfies $\zeta_a \cap X'_a$. Now let $\zeta'_a$ be the set of all $(\rowa,\cola) \in X''_a$ such that there is some $(\rowa',\cola') \in \pi'_a \cup \pi_{\A_\leq \cup H}$ where $|\rowa-\rowa'|=|\cola-\cola'|=1$. Notice that $\zeta_a \cap X''_a \subseteq \zeta'_a$.
Remember that $\size{X''_a} \geq 9^2$ and note that $\zeta'_a$ does not have more than $4$ restrictions on each row and on each column. Then, applying again Corollary~\ref{cor:constraint-(n,4)}, there is a permutation $\pi''_a$ over $X''_a$ that satisfies $\zeta'_a \cap X''_a$. By definition of $X'_a$ and $X''_a$ we have that $\pi_a = \pi'_a \cup \pi''_a$ is a permutation over $X_a$ that satisfies both $\zeta_a \cap X'_a$ and $\zeta'_a$. Further, since $\zeta_a \cap X''_a \subseteq \zeta'_a$, we have that $\pi_a$ satisfies $\zeta_a$. Also, by definition of $\zeta'_a$, we further have that there are no two $(\rowa,\cola), (\rowa',\cola') \in \pi_a$ with $|\rowa-\rowa'| = |\cola-\cola'| = 1$.
\end{proof}

We therefore define the new permutation $\pi_{\A_\leq \cup H \cup \set a} = \pi_a \cup \pi_{\A_\leq \cup H}$; and $\lambda_{\A_\leq \cup H \cup \set a}(\rowa,\cola) = a$ for all $(\rowa,\cola) \in X_a$, and $\lambda_{\A_\leq \cup H \cup \set a}(\rowa,\cola) = \lambda_{\A_\leq \cup H}(\rowa,\cola)$ otherwise. Note that it satisfies $R$ since for all new positions  $(\rowa,\cola)$ added there is no other position $(\rowa',\cola')$ such that $|\rowa-\rowa'|=|\cola-\cola'|=1$.

Finally, the desired labeled permutation is $(\pi_{\A_\leq \cup \A_>}, \lambda_{\A_\leq \cup \A_>})$.

\medskip

$[\Leftarrow]$ Suppose  that $pk(\lang{\+A'_1}) ~\cap~ pk(\lang{\+A'_2}) = \emptyset$; we show that there is no labeled permutation $(\pi, \lambda)$ verifying the restrictions imposed by the problem. 

First, we show that there cannot be a solution whose every label appears at most $17$ times. By means of contradiction, suppose $(\pi, \lambda)$ is such a solution. Then, the algorithm can guess $(\pi, \lambda)$, and we would then have that, by construction, $pk(\lang{\+A'_1}) = pk(\lang{\+A'_2}) \neq \emptyset$, which is in contradiction of our hypothesis.

Now suppose that $(\pi,\lambda)$ is a solution, where $\A_> \neq \emptyset$ is the set of letters that appear more than $17$ times in  $(\pi, \lambda)$. Let $(\pi'',\lambda'')$ be the replacement in $(\pi, \lambda)$ of every label from $\A_>$  with $\Box$. We can now build $(\pi',\lambda')$ from $(\pi'',\lambda'')$, by replacing each block containing only elements $\Box$ (and is maximal in size with respect to this property) with only one element $\Box$. For example, if $(\pi'',\lambda'')$ is as depicted in Figure~\ref{fig:claim-polyboxes}-$b$, we produce $(\pi',\lambda')$ by removing two rows and two columns, ending up with the labeled permutation of Figure~\ref{fig:claim-polyboxes}-$c$.

Let $\rex_1,\rex_2,\+A_>, \+A_1', \+A'_2$ be defined as before from $(\pi',\lambda')$. It is clear that $(\pi, \lambda)_\suca$ is in $L(\rex_1)$ and $(\pi, \lambda)_\suca$ in $L(\rex_2)$. It is therefore true that $(\pi, \lambda)_\suca \in \lang{\+A_1'}$ and $(\pi, \lambda)_\sucb \in \lang{\+A_2'}$. Since $(\pi, \lambda)$ is a labeled permutation we have that $pk((\pi, \lambda)_\suca) = pk((\pi, \lambda)_\sucb)$. Then, $pk(\lang{\+A'_1}) ~\cap~ pk(\lang{\+A'_2}) \neq \emptyset$, which is in contradiction with our hypothesis. 
Therefore, if $pk(\lang{\+A'_1}) ~\cap~ pk(\lang{\+A'_2}) = \emptyset$, there cannot be a solution $(\pi, \lambda)$ to the RLP instance.
\end{proof}

\subsection{Satisfiability for $\FOtwo$}
We now show that there is a \nexptime reduction from the satisfiability problem for $\FOtwo(\suca, \sucb)$ into the RLP problem. This, combined with the fact that RLP is in \np (Proposition~\ref{prop:rlpp:NP}), concludes the proof of Theorem~\ref{th:FOtwo(suca,sucb):nexptime},
 showing that the satisfiability problem for $\FOtwo(\suca,\sucb)$ is in \nexptime; hence it is \nexptime-complete~\cite{EtessamiVW02}.

Before going into the reduction, we show that the satisfaction of a formula in a valued permutation depends solely on its sets of fingerprints, plus some summary information. This summary information says, for every possible valuation $S$, how many times $S$ appears in $(\pi,\sigma)$ (counting up to a threshold of 3).
In the reduction from $\FOtwo(\suca,\sucb)$ into RLP we guess the  summary information and set of fingerprints of a minimal model (bounded by Propositions~\ref{prop:small-blocks} and \ref{prop:few-fingerprints}), translating the formula into a RLP instance.

\begin{definition}
  Given a set of fingerprints $X$, let $\hat X$ be the set of all the valuations in $X$, that is $\hat X = \set{\tau_\suca[i] \mid \tau \in X, 1 \leq  i \leq |\tau_\suca| \land \tau_\suca[i] \neq \bot}$.
Given a set of fingerprints $X$ and disjoint sets of valuations $V_1,V_2, V_3 \subseteq \hat X$, \defstyle{a valued permutation $(\pi, \sigma)$ over $X, V_1, V_2, V_3$} is any valued permutation such that 
 $X$ is the set of fingerprints of $(\pi,\sigma)$, and
 for every $1 \leq i \leq 3$, $V_i$ is the set of valuations that appear exactly $i$ times in $(\pi,\sigma)$.
 We also say that $X, V_1, V_2, V_3$ is the \defstyle{summary} of $(\pi, \sigma)$.
\end{definition}

Given a formula $\forall y . \psi$ where $\psi$ is a quantifier-free formula, and a valued permutation $(\pi,\sigma)$ with a block $B$, whether all the elements from $B$ verify $\forall y . \psi$ or not, depends only on: the summary of $(\pi,\sigma)$, and
  the fingerprint of $B$.
Similarly for formulas $\exists y . \psi$. Moreover, we can test this in polynomial time.

\newcommand{\model}[3]{(#1) \,{#2}\, (#3)}
We introduce the concept of a fingerprint being \defstyle{consistent} with a formula $\forall x \forall y . \chi$ [resp.\ $\forall x \exists y . \psi$] and a summary.\footnote{In fact, we do not need the set of fingerprints $X$ to define this notion but just the set of valuations $\hat X$, and it is therefore defined over $\hat X$.} 
These are the necessary and sufficient conditions to ensure that every element of a block with such fingerprint in a valued permutation over such summary  satisfies $\forall y . \chi$ [resp.\ $\exists y . \psi$].

\smallskip

First, note that for any quantifier-free formula $\psi$ of $\FOtwo(\suca,\sucb)$, the validity of $(\pi,\sigma) \models_{\mu} \psi$ only depends on: $S$, $S'$ and $t$, where: $S=\sigma(\mu(x))$, $S'=\sigma(\mu(y))$, and $t=\ntype{\pi}{\mu(x)}{\mu(y)}$. We will then write $\model{S}{t}{S'} \models \psi$, to denote that $\psi$ holds in any model that assigns $S$ to $x$, $S'$ to $y$ and so that the neighborhood type between $x$ and $y$ is $t$. Notice that we can decide $\model{S}{t}{S'} \models \psi$ in polynomial time. For example, if $\psi = {x {\rightarrow} y} \land (a(x) \lor \lnot b(y))$, we have $\model{\set{b}}{\nearrow}{\set{a,c}} \models \psi$ but $\model{\set{a}}{\downarrow}{\set{a,b}} \not\models \psi$.

The formal definition of consistency is given next.
\begin{definition}\label{def:consistency}r
Let $\tau$ be a fingerprint, and $m = |\tau_\suca|$. Given a set of valuations $Y \subseteq 2^{\vars \varphi}$ and three disjoint sets $V_1, V_2, V_3 \subseteq Y$, we say that  $\tau$ is \defstyle{consistent with a universal formula} $\forall y . \chi$ and $V_1, V_2, V_3, Y$ if all of the following conditions hold:
\begin{enumerate*}
\small
  \item \label{def:consistency:forall:1}
For every $1 < i < m$, we have  $\model{\tau_\suca[i]}{\bullet}{\tau_\suca[i]} \models \chi$.
  \item \label{def:consistency:forall:2}

If $\tau_\suca[1] \neq \bot$, $\model{\tau_\suca[2]}{\leftarrow}{\tau_\suca[1]} \models \chi$.
  \item \label{def:consistency:forall:3}
If $\tau_\suca[m] \neq \bot$,
$
        \model{\tau_\suca[m-1]}{\rightarrow}{\tau_\suca[m]} \models \chi$.
  \item \label{def:consistency:forall:4}
For every $1 < i, i+1 < m$, if $\type(\tau)=\searrow$,
$
\model{\tau_\suca[i]}{\searrow}{\tau_\suca[i+1]} \models \chi$, 
        $\model{\tau_\suca[i+1]}{\nwarrow}{\tau_\suca[i]} \models \chi$,
 otherwise, if $\type(\tau) = \nearrow$,
$
\model{\tau_\suca[i]}{\nearrow}{\tau_\suca[i+1]} \models \chi$, $ \model{\tau_\suca[i+1]}{\swarrow}{\tau_\suca[i]} \models \chi$.
%


  \item \label{def:consistency:forall:5}
For every $1 \leq i-1, i, i+1 \leq m$,
    \begin{itemize*}
    \item for every $g \in Y \setminus \set{\tau_\suca[i-1], \tau_\suca[i],
        \tau_\suca[i+1]}$ then $\model{\tau_\suca[i]}{\infty}{g} \models \chi$,
   \item   for every  $g \in \set{\tau_\suca[i-1], \tau_\suca[i], \tau_\suca[i+1]}$, if $g \not\in V_j$ for 
$j = \size{\set{ t \in \set{1,0,-1} \mid \tau_\suca[i+t] = g}}$ then
$
  \model{\tau_\suca[i]}{\infty}{g} \models \chi$.
    \end{itemize*}
\item \label{def:consistency:forall:6}
Idem to items 2 and 3, but replacing $\suca$ with $\sucb$, $\leftarrow$ with $\uparrow$, and $\tau_\suca$ with $\tau_\sucb$.
\end{enumerate*}

We say that $\tau$ is \defstyle{consistent with an existential formula} $\exists y . \psi$ if for every $1 < i < m$ either
\begin{enumerate*}
\small
\item \label{def:consistency:exists:1}
  \begin{enumerate*}
  \item\label{def:consistency:exists:1:a} $\model{\tau_\suca[i]}{\bullet}{\tau_\suca[i]}
    \models \psi$, 
  \item\label{def:consistency:exists:1:b} $i+1 = m$ and $\tau_\suca[i+1] \neq \bot$ and
    $\model{\tau_\suca[i]}{\rightarrow}{\tau_\suca[i+1]}
    \models \psi$, or
  \item\label{def:consistency:exists:1:c} $i = 2$, $\tau_\suca[1] \neq \bot$ and
    $\model{\tau_\suca[2]}{\leftarrow}{\tau_\suca[1]}
    \models \psi$, 
  \end{enumerate*}
\item\label{def:consistency:exists:2} $\type(\tau)=\searrow$, and either
\begin{enumerate*}
    \item\label{def:consistency:exists:2:a} $1 < i, i+1 < m$ and $\model{\tau_\suca[i]}{\searrow}{\tau_\suca[i+1]} \models \psi$, 
    \item\label{def:consistency:exists:2:b} $1 < i-1, i < m$ and $\model{\tau_\suca[i]}{\nwarrow}{\tau_\suca[i-1]} \models \psi$, 
    \item\label{def:consistency:exists:2:c} $1 \leq i-1, i, i+1 \leq m$ and there is some $g \in Y \setminus \set{\tau_\suca[i-1],
        \tau_\suca[i], \tau_\suca[i+1]}$, such that
$
        \model{\tau_\suca[i]}{\infty}{g} \models \psi 
$, or
    \item\label{def:consistency:exists:2:d} $1 \leq i-1, i, i+1 \leq m$ and there is $g \in \set{\tau_\suca[i-1], \tau_\suca[i],
        \tau_\suca[i+1]}$ with $g \not\in V_j$ for $j = \size{\set{ t \in \set{1,0,-1} \mid \tau_\suca[i+t] = g}}$ such that
      $\model{\tau_\suca[i]}{\infty}{g} \models \psi$, 
    \end{enumerate*}
  \item\label{def:consistency:exists:3} $\type(\tau) = \nearrow$, and some condition as the ones in item \ref{def:consistency:exists:2} holds, where $\searrow$ and $\nwarrow$ are replaced with $\nearrow$ and $\swarrow$,  or
  \item\label{def:consistency:exists:4} Idem as condition $1$,  replacing $\suca$ with $\sucb$, $\leftarrow$ with $\uparrow$, and $\tau_\suca$ with $\tau_\sucb$.
\end{enumerate*}

Finally, we say that a fingerprint $\tau$ is consistent with a formula in Scott normal form $\varphi = \forall x \forall y . \chi \land \bigwedge_i \forall x \exists y . \psi_i$  and sets $Y,V_1, V_2, V_3$ if it is consistent with $\forall y . \chi$ and with $\exists y . \psi_i$ for all $i$.
\end{definition}
The following Lemmas follow straightforward from the previous definitions.

\begin{lemma}\label{lem:consistency:1}
  For every valued permutation $(\pi,\sigma)$ over $X, V_1,V_2,V_3$, with a maximal block $B$, and for every formula $\forall y . \chi$ where $\chi$ is quantifier-free, we have that all the elements from $B$ verify $\forall y . \chi$ if and only if $\xi(B)$ is consistent with $\forall y . \chi$ and $\hat X, V_1,V_2,V_3$.
\end{lemma}

\begin{lemma}\label{lem:consistency:2}
  For every valued permutation $(\pi,\sigma)$ over $X, V_1,V_2,V_3$, with a maximal block $B$, for every formula $\exists y . \psi$  where $\psi$ is quantifier-free, we have that all the elements from $B$ verify $\exists y . \psi$ if and only if $\xi(B)$ is consistent with $\exists y . \psi$ and $\hat X, V_1,V_2,V_3$.
\end{lemma}

\begin{remark}\label{rem:consistency-polynomial}
  Note that the property of consistency of Definition~\ref{def:consistency} can be checked in polynomial time. 
\end{remark}


\begin{lemma}
There is a \nexptime reduction from the satisfiability problem for $\FOtwo(\suca,\sucb)$ into the RLP problem.
\end{lemma}
\begin{proof}
Let $\varphi \in \FOtwo(\suca,\sucb)$ be in Scott normal form, $\varphi = \forall x \forall y . \chi \land \bigwedge_i \forall x \exists y . \psi_i$.  By Proposition~\ref{prop:few-fingerprints}, there are at most an exponential number of different fingerprints, and by Proposition~\ref{prop:small-blocks} each one of them is at most of exponential size. The algorithm guesses the set $X$ of all fingerprints needed in a minimal valued permutation that satisfies $\varphi$, and summary sets $V_1,V_2,V_3 \subseteq \hat X$.  The algorithm checks that for every $\tau \in X$,  $\tau$ is consistent with $\varphi$ and $\hat X, V_1, V_2, V_3$.

We define the language $\+L_1 \subseteq X^*$ [resp.\ $\+L_2 \subseteq X^*$] of all words $w \in X^*$ such that
\begin{itemize*}
  \item the first element of $(w[1])_\suca$ [resp.\ of $(w[1])_\sucb$] is $\bot$, and the last element of $(w[|w|])_\suca$ [resp.\ of $(w[|w|])_\sucb$] is $\bot$,
  \item for every $1 \leq i < |w|$ such that  $(w[i])_\suca = a_1 \dotsb a_{n-1} a_n$ [resp.\ $(w[i])_\sucb = a_1 \dotsb a_{n-1} a_n$] and $(w[i+1])_\suca = b_1 b_2 \dotsb b_m$ [resp.\ $(w[i+1])_\sucb = b_1 b_2 \dotsb b_m$] we have $a_{n-1} a_n = b_1 b_2$,
  \item all the elements of $X$ appear in $w$,
  \item for every $1 \leq i \leq 3$, $V_i$ is the set of valuations that appear exactly $i$ times in $\hat w$, where
$\hat w = \hat w_1 \dotsb \hat w_{|w|}$
and $
\hat w_i = (w[i])_\suca[2] \dotsb (w[i])_\suca[m_i -1]
$
 for $m_i = |(w[i])_\suca|$.
\end{itemize*}
It is immediate that $\+L_1$ and $\+L_2$ are regular languages, and that they can be defined by two NFA that can be built in polynomial time in the size of $X$.

Finally, we define the label restrictions, avoiding having two blocks that actually define a bigger block (because these blocks are supposed to be maximal). Let $R$ be the set of all triples $(\tau,d,\tau')$ such that $\tau, \tau' \in X$ and either
\begin{itemize*}
  \item $d = \searrow$,  $\type(\tau) \neq \nearrow$, $\type(\tau') \neq \nearrow$, or
  \item  $d = \nearrow$, $\type(\tau)\neq \searrow$, $\type(\tau') \neq \searrow$.
\end{itemize*}
We reduced the satisfiability problem into the RLP problem for $(X,R,\+L_1, \+L_2)$, concluding the proof.
\begin{claim}\label{cl:RLP<=>SAT}
  The RLP instance $(X,R,\+L_1, \+L_2)$ has a positive solution if{f} the formula $\varphi$ is satisfiable.\qedhere
\end{claim}
\end{proof}

\section{Conclusion}
Our work shows that the following combinatorial problem is at the core of the satisfiability for $\FOtwo(\suca,\sucb)$ and of the RLP problem, and is decidable in \np. Given two regular languages $\+L, \+L' \subseteq \A^*$, is there a word $a_1 \dotsb a_n \in \+L$ and a permutation $p: [n] \to [n]$ so that $a_{p(1)} \dotsb a_{p(n)} \in \+L'$ and $|p(i+1) - p(i)|>1$ for all $i \in [n]$? 

A natural question left open is whether this decidability result can be extended to $\FOtwo$ with $k$ successor relations over finite linear orders is decidable, for arbitrary $k$ (or at least for $k=3$).


%



\bibliography{short,biblio}

 \appendix

 \section{Missing proofs}

\begin{proof}[Proof of Lemma~\ref{lem:consistency:1}]
\hfill

[$\Rightarrow$]
Suppose that every element $(\rowa,\cola)\in B \cap \pi$ verifies $(\pi,\sigma) \models_{[x \mapsto (\rowa,\cola)]} \forall y . \chi$. We show that $\tau = \xi(B)$ is consistent with $\forall y . \chi$ and $\hat X, V_1,V_2,V_3$.

Condition \ref{def:consistency:forall:1} is met, since in particular $(\pi,\sigma) \models_{[x \mapsto (\rowa,\cola), y \mapsto (\rowa,\cola)]} \chi$ for every $(\rowa,\cola) \in B \cap \pi$; this means that $\model{\sigma(\rowa,\cola)}{\bullet}{\sigma(\rowa,\cola)} \models \chi$, which implies condition \ref{def:consistency:forall:1}.

If $i$ is the smallest column element of $B$ and $i>1$, it means that $\tau_\suca[1] \neq \bot$, in fact $\tau_\suca[1] = \sigma(\nth,i-1)$. Since for some $(\rowa,i), (\rowa',i-1) \in \pi$,  $(\pi,\sigma) \models_{[x \mapsto (\rowa,i), y \mapsto (\rowa',i-1)]} \chi$,  and since $B$ is maximal, we have that $|\rowa - \rowa'| >1$ and hence that 
$\model{\sigma(\rowa,i)}{\leftarrow}{\sigma(\rowa',i-1)} \models \chi$. As $\sigma(\rowa',i-1)=\tau_\suca[1]$ and $\sigma(\rowa,i)=\tau_\suca[2]$, condition \ref{def:consistency:forall:2} is met. Condition \ref{def:consistency:forall:3} is similar.

Suppose $\type(\tau) = {\searrow}$ and $1 < i, i+1 < |\tau_\suca|$. This means that there are  $(\rowa,\cola), (\rowa+1, \cola+1) \in B \cap \pi$ with $\sigma(\rowa,\cola) = \tau_\suca[i]$, $\sigma(\rowa+1, \cola+1) = \tau_\suca[i+1]$. Since $(\pi,\sigma) \models_{[x \mapsto (\rowa,\cola), y \mapsto (\rowa+1,\cola+1)]} \chi$ and $(\pi,\sigma) \models_{[y \mapsto (\rowa,\cola), x \mapsto (\rowa+1,\cola+1)]} \chi$ we have that $\model{\sigma(\rowa,\cola)}{\searrow}{\sigma(\rowa+1,\cola+1)} \models \chi$ and $\model{\sigma(\rowa+1,\cola+1)}{\nwarrow}{\sigma(\rowa,\cola)} \models \chi$. Hence, $\model{\tau_\suca[i]}{\searrow} {\tau_\suca[i+1]} \models \chi$ and $\model{\tau_\suca[i+1]}{\nwarrow}{\tau_\suca[i]} \models \chi$ and thus condition \ref{def:consistency:forall:4} holds. The proof for $\type(\tau) = {\nwarrow}$ is similar.

Suppose now that $1 \leq i-1, i, i+1 \leq |\tau_\suca|$. Let $g \in \hat X \setminus \set{\tau_\suca[i-1], \tau_\suca[i], \tau_\suca[i+1]}$. This means that there is some $(\rowa,\cola) \in B \cap \pi$ with $\sigma(\rowa,\cola)=\tau_\suca[i]$ and 
some $(\rowa',\cola') \in \pi$ with $\sigma(\rowa',\cola')=g$, $\ntype{\pi}{(\rowa,\cola)}{(\rowa',\cola')} = \infty$. Since $(\pi,\sigma) \models_{[x \mapsto (\rowa,\cola), y \mapsto (\rowa',\cola')]} \chi$, we have that $\model{\sigma(\rowa,\cola)}{\infty}{\sigma(\rowa',\cola')} \models \chi$ and hence the first part of condition \ref{def:consistency:forall:5} is met. On the other hand, if $g \in \set{\tau_\suca[i-1], \tau_\suca[i], \tau_\suca[i+1]}$ and $g \not\in V_j$ for $j = \size{\set{ t \in \set{1,0,-1} \mid \tau_\suca[i+t] = g}}$, this means that there must be necessarily some other position with valuation $g$. That is, there is some $(\rowa',\cola') \in \pi$ with $\sigma(\rowa',\cola')=g$, $\ntype{\pi}{(\rowa,\cola)}{(\rowa',\cola')} = \infty$. Then, the same reasoning as before applies to show that 
the second part of condition \ref{def:consistency:forall:5} is met.

\medskip

[$\Leftarrow$] 
Suppose that $\tau = \xi(B)$ is consistent with $\forall y . \chi$ and $\hat X, V_1,V_2,V_3$. Let us show that  every element $(\rowa,\cola)\in B \cap \pi$ verifies $(\pi,\sigma) \models_{[x \mapsto (\rowa,\cola)]} \forall y . \chi$. 

Let $(\rowa,\cola) \in B \cap \pi$ and $(\rowa',\cola') \in \pi$. If $\ntype{\pi}{(\rowa,\cola)}{(\rowa',\cola')} = \bullet$ then $(\pi,\sigma) \models_{[x \mapsto (\rowa,\cola), y \mapsto (\rowa',\cola')]} \chi$ by condition \ref{def:consistency:forall:1}. 
If $\ntype{\pi}{(\rowa,\cola)}{(\rowa',\cola')} = {\leftarrow}$, then $(\pi,\sigma) \models_{[x \mapsto (\rowa,\cola), y \mapsto (\rowa',\cola')]} \chi$ by condition \ref{def:consistency:forall:2}.
If $\ntype{\pi}{(\rowa,\cola)}{(\rowa',\cola')} = {\rightarrow}$, then $(\pi,\sigma) \models_{[x \mapsto (\rowa,\cola), y \mapsto (\rowa',\cola')]} \chi$ by condition \ref{def:consistency:forall:3}.
If $\ntype{\pi}{(\rowa,\cola)}{(\rowa',\cola')} \in \set{ \nearrow, \nwarrow, \searrow,\swarrow}$, then $(\pi,\sigma) \models_{[x \mapsto (\rowa,\cola), y \mapsto (\rowa',\cola')]} \chi$ by condition \ref{def:consistency:forall:4}. 
And if $\ntype{\pi}{(\rowa,\cola)}{(\rowa',\cola')} = \infty$, then $(\pi,\sigma) \models_{[x \mapsto (\rowa,\cola), y \mapsto (\rowa',\cola')]} \chi$ by condition \ref{def:consistency:forall:5}.  Hence, every element $(\rowa,\cola)\in B \cap \pi$ verifies $(\pi,\sigma) \models_{[x \mapsto (\rowa,\cola)]} \forall y . \chi$. 
\end{proof}

\begin{proof}[Proof of Lemma~\ref{lem:consistency:2}]

[$\Rightarrow$]
Suppose that every element $(\rowa,\cola)\in B \cap \pi$ verifies $(\pi,\sigma) \models_{[x \mapsto (\rowa,\cola)]} \exists y . \psi$. We show that $\tau = \xi(B)$ is consistent with $\exists y . \psi$ and $\hat X, V_1,V_2,V_3$. 

Let $(\rowa,\cola) \in B \cap \pi$, then there must be $(\rowa',\cola') \in \pi$ such that $(\pi,\sigma) \models_{[x \mapsto (\rowa,\cola), y \mapsto (\rowa',\cola')]} \psi$.
If $\ntype{\pi}{(\cola,\rowa)}{(\rowa',\cola')} = {\bullet}$ then condition  \ref{def:consistency:exists:1:a} is met. 
If $\ntype{\pi}{(\cola,\rowa)}{(\rowa',\cola')} = {\rightarrow}$, then condition \ref{def:consistency:exists:1:b} is met.
If $\ntype{\pi}{(\cola,\rowa)}{(\rowa',\cola')} = {\leftarrow}$, then condition \ref{def:consistency:exists:1:c} is met.
If $\ntype{\pi}{(\cola,\rowa)}{(\rowa',\cola')} = {\searrow}$, then condition \ref{def:consistency:exists:2:a} is met.
If $\ntype{\pi}{(\cola,\rowa)}{(\rowa',\cola')} = {\nwarrow}$, then condition \ref{def:consistency:exists:2:b} is met.
If $\ntype{\pi}{(\cola,\rowa)}{(\rowa',\cola')} = {\infty}$, then either condition \ref{def:consistency:exists:2:c} or \ref{def:consistency:exists:2:d} is met.
If $\ntype{\pi}{(\cola,\rowa)}{(\rowa',\cola')} \in \set{\nearrow,\swarrow}$, then condition \ref{def:consistency:exists:3} is met.
If $\ntype{\pi}{(\cola,\rowa)}{(\rowa',\cola')} \in \set{\uparrow,\downarrow}$, then condition \ref{def:consistency:exists:4} is met.
Thus, $\tau = \xi(B)$ is consistent with $\exists y . \psi$ and $\hat X, V_1,V_2,V_3$.

\medskip

[$\Leftarrow$]
Suppose that $\tau = \xi(B)$ is consistent with $\exists y . \psi$ and $\hat X, V_1,V_2,V_3$. Let us show that  every element $(\rowa,\cola)\in B \cap \pi$ verifies $(\pi,\sigma) \models_{[x \mapsto (\rowa,\cola)]} \exists y . \psi$. 

Let $(\rowa,\cola) \in B \cap \pi$. There must be some $1 < i < |\tau_\suca|$ such that $\tau_\suca[i] = \sigma(\rowa,\cola)$.
If condition  \ref{def:consistency:exists:1:a} holds, then
$(\pi,\sigma) \models_{[x \mapsto (\rowa,\cola), y \mapsto (\rowa,\cola)]} \psi$.
If condition  \ref{def:consistency:exists:1:b} holds, then
$(\pi,\sigma) \models_{[x \mapsto (\rowa,\cola), y \mapsto (\rowa',\cola+1)]} \psi$ for some $(\rowa',\cola+1) \in \pi$ with $|\rowa' - \rowa| > 1$.
If condition  \ref{def:consistency:exists:1:c} holds, then
$(\pi,\sigma) \models_{[x \mapsto (\rowa,\cola), y \mapsto (\rowa',\cola-1)]} \psi$ for some $(\rowa',\cola-1) \in \pi$ with $|\rowa' - \rowa| > 1$.
If condition  \ref{def:consistency:exists:2:a} holds, then
$(\pi,\sigma) \models_{[x \mapsto (\rowa,\cola), y \mapsto (\rowa+1,\cola+1)]} \psi$ for $(\rowa+1,\cola+1) \in \pi$.
If condition  \ref{def:consistency:exists:2:b} holds, then
$(\pi,\sigma) \models_{[x \mapsto (\rowa,\cola), y \mapsto (\rowa-1,\cola-1)]} \psi$ for $(\rowa-1,\cola-1) \in \pi$.
If condition  \ref{def:consistency:exists:2:c} or \ref{def:consistency:exists:2:d} holds, then
$(\pi,\sigma) \models_{[x \mapsto (\rowa,\cola), y \mapsto (\rowa-1,\cola-1)]} \psi$ for some $(\rowa',\cola') \in \pi$ with $\ntype{\pi}{(\rowa,\cola)}{(\rowa',\cola')} = \infty$.
A similar reasoning applies if \ref{def:consistency:exists:3} or \ref{def:consistency:exists:4} hold.
\end{proof}

\begin{proof}[Proof of Proposition~\ref{prop:small-blocks}]
We show the following statement: Any minimal valued permutation satisfying $\varphi \in \FOtwo(\suca,\sucb)$ is such that  every block is of size less or equal to \[3 \cdot 2^{4\size{\vars \varphi} +3}+2^{3(\size{\vars \varphi}+1)} + 3 \cdot 2^{\size{\vars \varphi}}.\]

\smallskip

Let $(\pi,\sigma)$ be a minimal valued permutation such that $(\pi,\sigma) \models \varphi$. Without loss of generality assume $\sigma : \pi \to 2^{\vars \varphi}$.

Suppose, by means of contradiction, that $B = [i,i+k] \times [i',i'+k]$ is a block of $\pi$ of size $k > 3 \cdot 2^{4 \size{\vars \varphi} +3}+2^{3(\size{\vars \varphi}+1)} + 3 \cdot 2^{\size{\vars \varphi}}$. Let us assume that $B$ is of type $\searrow$ (if it has type $\nearrow$ a symmetrical reasoning applies). For every $S \subseteq \vars \varphi$, choose any three elements from $\set{(\rowa,\cola) \in \pi \cap B \mid \sigma(\rowa,\cola)  = S}$, or, if there are less than three, all the elements; let $S_\pi$ be the set of these three (or less) elements. Now consider all the sub-blocks $B' \subsetneq B$ defined strictly between elements of $\bigcup_{S \subseteq \vars \varphi} S_\pi$, that is, such that $B' \cap \bigcup_{S \subseteq \vars \varphi} S_\pi = \emptyset$, consider the ones that are maximal with respect to inclusion. Since $\size{\bigcup_{S \subseteq \vars \varphi} S_\pi} \leq 3 \cdot 2^{\size{\vars \varphi}}$, there are at most $3 \cdot 2^{\size{\vars \varphi}}+1$ such sub-blocks having a total of at least $k - 3 \cdot 2^{\size{\vars \varphi}}$ elements. By the Pigeonhole Principle, since $k - 3 \cdot 2^{\size{\vars \varphi}} > 3 \cdot 2^{4\size{\vars \varphi} +3}+2^{3(\size{\vars \varphi}+1)} = (3 \cdot 2^{\size{\vars \varphi}}+1) \cdot 2^{3(\size{\vars \varphi}+1)}$, this means that there must be a sub-block with at least $2^{3(\size{\vars \varphi}+1)}+1$ elements; suppose it is 
\[B'=([i+\ell,i+\ell+k'] \times [i'+\ell,i'+\ell+k']) \subseteq B,\] 
where $k' > 2^{3(\size{\vars\varphi} +1)}$ and $\ell+k' \leq k$. There must be two distinct elements $(\rowa,\cola),(\rowa',\cola') \in \pi \cap B'$ with 
\begin{align*}
  \begin{split}
    \sigma(\rowa,\cola) &= \sigma(\rowa',\cola')\text{,} \\
    \sigma(\rowa-1,\nth) &= \sigma(\rowa'-1,\nth)\text{, and} \\
    \sigma(\rowa+1,\nth)&=\sigma(\rowa'+1,\nth).
  \end{split}
\tag{$\dag$}
\label{eq:three-elems-eq}
\end{align*}
 Without any loss of generality, suppose that $\rowa<\rowa'$ and $\cola<\cola'$ (the same argument works if $\rowa<\rowa'$ and $\cola'<\cola$). Let $\pi' = \pi \setminus ([\rowa+1,\rowa'] \times [\cola+1,\cola'])$. It is easy to verify  that $(\pi',\sigma') \models \varphi$, where $\sigma' = \sigma|_{\pi'}$.
\begin{claim}\label{cl:smallblocks}
  $(\pi',\sigma') \models \varphi$
\end{claim}
\begin{proof}
Let $\varphi = \forall x \forall y . \chi ~\land~ \bigwedge_j \forall x \exists y . \psi_j$. 
 For any $(\rowb,\colb), (\rowb',\colb') \in \pi$ such that $(\rowb,\colb) \in \pi'$ and $(\pi,\sigma) \models_{[x \mapsto (\rowb,\colb), y \mapsto (\rowb',\colb')]} \psi_j$, we can find $(\rowb'',\colb'') \in \pi'$ such that $\sigma(\rowb',\colb') = \sigma'(\rowb'',\colb'')$ and such that 
$\ntype{\pi'}{(\rowb'',\colb'')}{(\rowb,\colb)}=
\ntype{\pi}{(\rowb',\colb')}{(\rowb,\colb)}$. This is because, if the neighborhood type is not $\infty$ then it is taken care of by \eqref{eq:three-elems-eq}. Otherwise, it is witnessed by one of the elements from $S_\pi$, for $\sigma(\rowb',\colb')=S$. Therefore, $(\pi',\sigma')$ verifies $\forall x \exists y . \psi_j$ for every $j$.

On the other hand, if $(\rowb, \colb), (\rowb',\colb') \in \pi'$ where $(\pi,\sigma) \models_{[x \mapsto (\rowb, \colb), y \mapsto (\rowb', \colb')]} \chi$, we have that, if $\ntype{\pi'}{(\rowb,\colb)}{(\rowb',\colb')} = 
\ntype{\pi}{(\rowb,\colb)}{(\rowb',\colb')}$, then $(\pi',\sigma') \models_{[ x \mapsto (\rowb,\colb), y \mapsto (\rowb',\colb') ]} \chi$. Otherwise, we necessarily have that $(\rowb,\colb) = (\rowa,\cola)$, $(\rowb',\colb')=(\rowa'+1,\cola'+1)$ (or vice-versa). But since $\sigma(\rowa+1,\cola+1) = \sigma(\rowa'+1,\cola'+1)$ and $(\pi,\sigma) \models_{[ x \mapsto (\rowa,\cola), y \mapsto (\rowa+1, \cola+1) ]} \chi$, then $(\pi',\sigma') \models_{[ x \mapsto (\rowb,\colb), y \mapsto (\rowb', \colb') ]} \chi$, because $\ntype{\pi'}{(\rowb,\colb)}{(\rowb',\colb')} = 
\ntype{\pi}{(\rowa,\cola)}{(\rowa+1,\cola+1)} = {\searrow}$.
\end{proof}

Since $\pi'$ is smaller than $\pi$, and $(\pi',\sigma')$ verifies $\varphi$, it cannot be that $\pi$ is minimal in size, which is in contradiction with our hypothesis. Therefore, $\pi$ does not have blocks of size bigger than $3 \cdot 2^{4\size{\vars \varphi} +3}+2^{3(\size{\vars \varphi}+1)} + 3 \cdot 2^{\size{\vars \varphi}}$. 
\end{proof}

\begin{proof}[Proof of Proposition~\ref{prop:few-fingerprints}]
Let $(\pi,\sigma) \models \varphi$ be such that $(\pi,\sigma)$ is minimal in size. Without any loss of generality we assume that $\sigma : \pi \to 2^{\vars \varphi}$.
 We show how to build another minimal valued permutation from $(\pi,\sigma)$ with a number block fingerprints that is bounded by an exponential function on $|\varphi|$.

Let $\pi = B_1 \cup \dotsb \cup B_N$, where each $B_i$ is a maximal block. We are going to mark some blocks (only exponentially many), and use only these  to build a new valued permutation satisfying $\varphi$. For each $S \subseteq \vars \varphi$, let $S_\pi \subseteq \pi$ be a set of 4 elements of $\set{(\rowa,\cola) \in \pi \mid \sigma(\rowa,\cola) = S}$ ---if the set has less than 4 elements, then let $S_\pi$ be all of them. For every $S \subseteq \vars \varphi$ and every $(\rowa,\cola) \in S_\pi$, mark the block $B_k$ such that $(\rowa,\cola) \in B_k$ with a color red. It follows that we end up with at most $4 \cdot 2^{\size{\vars \varphi}}$ blocks marked with red. On the other hand, for every $i\in [N]$, let us define $h_i = (d,b^\suca_{-1},b^\suca_{+1},b^\sucb_{-1},b^\sucb_{+1})$, where
 $\xi(B_i) = (d, b^\suca_{-1},b^\suca_{+1},b^\sucb_{-1},b^\sucb_{+1}, a_0\dotsb a_{t})$.
For each $h \in \set{ h_i \mid i \in [N] \text{ s.t.\ $B_i$ is not marked red}}$ let us mark with color green one block $B_i$ such that $h_i = h$ and such that $B_i$ is not marked red. Since there are no more than $3 \cdot 2^{4 \size{\vars \varphi}}$ different $h_i$'s, there are only an exponential number of blocks marked with colors red or green.

Now, let $B_m$ be a block with a fingerprint different from all the fingerprints of the marked blocks. There must be a block $B_n$ marked with green, so that $h_n = h_m$. Let $(\pi',\sigma')$ be the result of replacing of $B_m$ with $B_n$ in $(\pi,\sigma)$, in the expected way. 
We then have that $(\pi',\sigma') \models \varphi$.
\begin{claim}\label{cl:fewblocks}
  $(\pi',\sigma') \models \varphi$.
\end{claim}
\begin{proof}
Let $\varphi = \forall x \forall y. \chi ~\land~ \bigwedge_i \forall x \exists y . \psi_i$. 

We first show that $\forall x \forall y . \chi$ holds in $(\pi',\sigma')$.
Let $(\rowa,\cola), (\rowa',\cola')\in \pi'$.
 \begin{itemize*}
 \item If $(\rowa,\cola)$ is in the edited block of the permutation (\ie, in the new copy of the fingerprint of $B_n$).
   \begin{itemize*}
   \item If $\ntype{\pi}{(\rowa,\cola)}{(\rowa',\cola')} \neq \infty$, then both elements are described inside
     the fingerprint $\xi(B_n)$. Hence, there are two elements in the
     block $B_n$ of $(\pi,\sigma)$ with the same neighborhood type  and with the
     same valuations. Since for those two we have that $\chi$ holds,
     then it holds for $(\rowa,\cola), (\rowa',\cola')$ as well. Hence, $(\pi',\sigma')
     \models_{[x \mapsto (\rowa,\cola), y \mapsto (\rowa',\cola')]} \chi$.
   \item Otherwise, suppose $\ntype{\pi}{(\rowa,\cola)}{(\rowa',\cola')} = \infty$. Note that in the red colored blocks of
     $(\pi,\sigma)$ there must be at least 4 elements with value
     $\sigma'(\rowa,\cola)$. By the same reason there must be 4 elements with
     value $\sigma'(\rowa',\cola')$. Since marked blocks are preserved, there must be $(\rowb,l), (\rowb',\colb') \in
     \pi$ such that $\ntype{\pi}{(\rowb,\colb)}{(\rowb',\colb')}=\infty$, and $\sigma(\rowb,\colb) =
     \sigma'(\rowa,\cola)$, $\sigma(\rowb',\colb') = \sigma'(\rowa',\cola')$. Since
     $(\pi,\sigma) \models_{[x \mapsto (\rowb,\colb), y \mapsto (\rowb',\colb')]}
     \chi$, we have $(\pi',\sigma') \models_{[x \mapsto (\rowa,\cola), y
       \mapsto (\rowa',\cola')]} \chi$.
   \end{itemize*}
\item If both $(\rowa,\cola)$ and $(\rowa',\cola')$ are in the part that was not modified, then there must clearly be $(\rowb,\colb), (\rowb',\colb') \in \pi$ with 
$\ntype{\pi}{(\rowb,\colb)}{(\rowb',\colb')} =
\ntype{\pi'}{(\rowa,\cola)}{(\rowa',\cola')}$, so that $\sigma(\rowb,\colb) = \sigma'(\rowa,\cola)$, $\sigma(\rowb',\colb') = \sigma'(\rowa',\cola')$. Hence, 
$(\pi',\sigma') \models_{[x \mapsto (\rowa,\cola), y \mapsto (\rowa',\cola')]} \chi$.
 \end{itemize*}

We now show that for every $i$ and $(\rowa,\cola) \in \pi'$,  $\exists y . \varphi_i$ holds.
 \begin{itemize*}
 \item If $(\rowa,\cola)$ is in the edited block of the permutation (\ie, in the new copy of the fingerprint of $B_n$), then there must be some $(\rowa',\cola')$ in $B_n$ such that $\sigma(\rowa',\cola') = \sigma'(\rowa,\cola)$. Since $(\pi,\sigma) \models_{[x \mapsto (\rowa',\cola')]} \exists y . \varphi_i$, there must be some $(\rowb',\colb') \in \pi$ such that $(\pi,\sigma) \models_{[x \mapsto (\rowa',\cola'), y \mapsto (\rowb',\colb')]} \varphi_i$.
   \begin{itemize*}
   \item If $\ntype{\pi}{(\rowb',\colb')}{(\rowa',\cola')} \neq \infty$, then it is described in the fingerprint of $B_n$ and therefore there must be some $(\rowb,\colb)$ such that 
$\ntype{\pi'}{(\rowb,\colb)}{(\rowa,\cola)}=\ntype{\pi}{(\rowb',\colb')}{(\rowa',\cola')}$, and $\sigma'(\rowb,\colb) = \sigma(\rowb',\colb')$, $\sigma'(\rowa,\cola) = \sigma(\rowa',\cola')$. Then, $(\pi',\sigma') \models_{[x \mapsto (\rowa,\cola), y \mapsto (\rowb,\colb)]} \varphi_i$ since $(\pi,\sigma) \models_{[x \mapsto (\rowa',\cola'), y \mapsto (\rowb',\colb')]} \varphi_i$.

   \item Otherwise, suppose that 
$\ntype{\pi}{(\rowb',\colb')}{(\rowa',\cola')}=\infty$. There must be 4 elements of $(\pi,\sigma)$ with the same value $\sigma(\rowb',\colb')$ inside blocks marked red. Then, there must be a block marked red in $\pi'$ containing some $(\rowb,\colb)$ with $\sigma'(\rowb,\colb) = \sigma(\rowb',\colb')$ and such that 
$\ntype{\pi'}{(\rowb,\colb)}{(\rowa,\cola)}=\infty$. Then, $(\pi',\sigma') \models_{[x \mapsto (\rowa,\cola), y \mapsto (\rowb,\colb)]} \varphi_i$ since $(\pi,\sigma) \models_{[x \mapsto (\rowa',\cola'), y \mapsto (\rowb',\colb')]} \varphi_i$.
   \end{itemize*}
\item If $(\rowa,\cola)$ is in the part that was not modified, then there is some $(\rowa',\cola') \in \pi$ with $\sigma(\rowa',\cola') = \sigma'(\rowa,\cola)$ and $(\rowa',\cola') \not\in B_m$. Since $(\pi,\sigma) \models_{[x \mapsto (\rowa',\cola')]} \exists y . \varphi_i$, there must be some $(\rowb',\colb') \in \pi$ such that $(\pi,\sigma) \models_{[x \mapsto (\rowa',\cola'), y \mapsto (\rowb',\colb')]} \varphi_i$.
  \begin{itemize*}
  \item If 
$\ntype{\pi}{(\rowb',\colb')}{(\rowa',\cola')} \neq \infty$, then $(\rowb',\colb')$ is described in the fingerprint of the block where $(\rowa',\cola')$ belongs. Since this fingerprint is preserved (because it is not $B_m$), there must be some $(\rowb,\colb)$ with $\sigma'(\rowb,\colb) = \sigma(\rowb',\colb')$ and such that 
$\ntype{\pi'}{(\rowa,\cola)}{(\rowb,\colb)} =
\ntype{\pi}{(\rowa',\cola')}{(\rowb',\colb')}$. Then, $(\pi,\sigma) \models_{[x \mapsto (\rowa',\cola'), y \mapsto (\rowb',\colb')]} \varphi_i$.
  \item Otherwise, if 
$\ntype{\pi}{(\rowb',\colb')}{(\rowa',\cola')}=\infty$, there must be 4 other positions with the same value as $(\rowb',\colb')$ in $(\pi,\sigma)$, in blocks marked red. Then, there must necessarily be some $(\rowb,\colb) \in \pi'$ such that $\sigma'(\rowb,\colb) = \sigma(\rowb',\colb')$ and such that 
$\ntype{\pi'}{(\rowb,\colb)}{(\rowa,\cola)}=\infty$.
Then, $(\pi,\sigma) \models_{[x \mapsto (\rowa',\cola'), y \mapsto (\rowb',\colb')]} \varphi_i$.
  \end{itemize*}
 \end{itemize*}
Therefore, $(\pi',\sigma') \models_{[x \mapsto (\rowa,\cola)]} \exists y . \varphi_i$.
\end{proof}

Note that since $B_m$ and $B_n$ have the same type, then the number of maximal blocks in $\pi'$ is the same as in $\pi$, and moreover the number of maximal blocks with fingerprints different from the marked fingerprints is decremented by one.
Note that in particular this means that $\size{B_m}=\size{B_n}$. Otherwise, if $\size{B_m} > \size{B_n}$, we would have that $(\pi',\sigma')$ is smaller than $(\pi,\sigma)$ which cannot be since $(\pi,\sigma)$ is minimal in size. Conversely, if $\size{B_m} < \size{B_n}$, we could have marked $B_m$ with green instead of $B_n$ and we arrive to the same contradiction.

We can repeat this operation with all the unmarked blocks ending up with a valued permutation whose every block has the fingerprint of a marked block. Since there are only exponentially many marked blocks, there are exponentially many fingerprints.
\end{proof}


\end{document}